\documentclass[envcountsame]{llncs}
\usepackage{amsmath,amssymb}
\usepackage{tikz-cd}
\usepackage{array}
\usepackage{tikz}
\usetikzlibrary{cd,automata}
\usetikzlibrary{fit}
\usetikzlibrary{arrows}
\usepackage{pgfplots}
\usepackage{hyperref}
\usepackage{mathtools}
\hypersetup{
    colorlinks,
    linkcolor={blue!70!black},
    citecolor={blue!70!black},
    urlcolor={blue!70!black}
}
\usepackage{multirow,bigdelim}
\usepackage{cleveref}
\usepackage{todonotes}
\usepackage{enumitem}

\usepackage[noend]{algpseudocode}
\usepackage{algorithm}
\usepackage{marvosym}

\newcommand{\noopsort}[1]{}

\newcommand{\N}{\mathbb{N}} % Natural numbers
\newcommand{\supp}{\mathtt{supp}} %Support
\newcommand{\lang}{\mathcal{L}}

\newcommand{\tables}{\mathsf{Table_{fin}}}

\newcommand{\LStar}{\ensuremath{\mathtt{L}^{\!\star}}}

\newcolumntype{C}[1]{>{\centering\arraybackslash$}m{#1}<{$}}
\newcolumntype{L}[1]{>{\raggedleft\arraybackslash$}m{#1}<{$}}
\newcolumntype{R}[1]{>{\raggedright\arraybackslash$}m{#1}<{$}}
\tikzset{trlab/.style={font=\scriptsize,inner sep=1pt,outer sep=1pt}}
\newenvironment{automaton}[1][]{
\begin{tikzpicture}[
state/.append style={inner sep=0pt,outer sep=0pt,minimum size=3.5ex},%
initial text={},->,>=stealth',shorten >=1pt,node distance=9ex,semithick,#1]
}
{\end{tikzpicture}%\end{center}
}
\renewcommand{\gets}{\leftarrow}
\renewcommand{\gets}{\leftarrow}

\newcommand{\R}{\mathbb{R}}

\newcommand{\Ps}{\mathcal{P}}
\newcommand{\Psf}{\Ps_f}

\newcommand{\eword}{\varepsilon}
\newcommand{\suffixes}{\mathsf{suffixes}}
\newcommand{\rank}{\mathsf{rank}}

\newcommand{\Semi}{\mathbb{S}}
\newcommand{\PID}{\mathbb{P}}
\newcommand{\free}{V}

\newcommand{\size}{\mathsf{size}}

\newcommand{\reach}{\mathsf{reach}}
\newcommand{\obs}{\mathsf{obs}}

\newcommand{\cs}{\mathsf{cs}}

\newcommand{\row}{\mathsf{row}}
\newcommand{\srow}{\mathsf{srow}}

\newcommand{\aut}{\mathcal{A}}

\newcolumntype{C}[1]{>{\centering\arraybackslash$}m{#1}<{$}}
\newcolumntype{L}[1]{>{\raggedleft\arraybackslash$}m{#1}<{$}}
\newcolumntype{R}[1]{>{\raggedright\arraybackslash$}m{#1}<{$}}

\author{Gerco van Heerdt\inst{1} \and Clemens Kupke\inst{2}(\Letter) \and Jurriaan Rot\inst{1,3} \and Alexandra Silva\inst{1}}
\institute{
	University College London, United Kingdom \\
	\email{\{gerco.heerdt,alexandra.silva\}@ucl.ac.uk} \and
	University of Strathclyde, United Kingdom \\
	\email{clemens.kupke@strath.ac.uk} \and
	Radboud University, The Netherlands \\
	\email{jrot@cs.ru.nl}}

\title{Learning Weighted Automata \\ over Principal Ideal Domains\thanks{%
	The research leading to this work was partially funded by the European Union's Horizon 2020 research and innovation programme under the ERC Starting Grant ProFoundNet (grant code 679127) and the Marie Sk\l{}odowska-Curie Grant Agreement No.\ 795119, by the EPSRC Standard Grant CLeVer (EP/S028641/1), and by GCHQ via the VeTSS grant ``Automated black-box verification of networking systems'' (4207703/RFA 15845).
}}

\begin{document}

\maketitle
\begin{abstract}
In this paper, we study active learning algorithms for weighted automata over a semiring. We show that a variant of Angluin's seminal \LStar\ algorithm works when the semiring is a principal ideal domain, but not for general semirings such as the natural numbers.
\end{abstract}

\section{Introduction}

Angluin's seminal \LStar\ algorithm~\cite{angluin1987} for active learning of deterministic automata (DFAs) has been successfully used in many verification tasks, including in automatically building formal models of chips in bank cards or finding bugs in network protocols (see~\cite{frits-cacm,DBLP:conf/dagstuhl/HowarS16} for a broad overview of successful applications of active learning). While DFAs are expressive enough to capture interesting properties, certain verification tasks require more expressive models. This motivated several researchers to extend \LStar\ to other types of automata, notably Mealy machines~\cite{vilar1996,Shahbaz:2009:IMM:1693345.1693363}, register automata~\cite{DBLP:conf/cav/IsbernerHS15,DBLP:journals/sigsoft/MuesHLKR16,AartsFKV15}, and nominal automata~\cite{DBLP:conf/popl/MoermanS0KS17}.

Weighted finite automata (WFAs) are an important model made popular due to their applicability in image processing and speech recognition tasks~\cite{culik1993image,DBLP:journals/corr/abs-cs-0503077}. The model is prevalent in other areas, including bioinformatics~\cite{DBLP:conf/icml/AllauzenMT08} and formal verification~\cite{DBLP:conf/atva/AminofKL11}. Passive learning algorithms and associated complexity results have appeared in the literature (see e.g.~\cite{DBLP:conf/nips/BalleM12} for an overview), whereas active learning has been less studied~\cite{DBLP:conf/cai/BalleM15,Bergadano}. Furthermore, the existing  learning algorithms, both passive and active, have been developed assuming the weights in the automaton are drawn from a field, such as the real numbers.\footnote{Balle and Mohri~\cite{DBLP:conf/cai/BalleM15} define WFAs generically over a semiring but then restrict to fields from Section 3 onwards as they present an overview of existing learning algorithms.} To the best of our knowledge, no learning algorithms, whether passive or active, have been developed for WFAs in which the weights are drawn from a general semiring.

In this paper, we explore {\em active learning} for WFAs over a general semiring. The main contributions of the paper are as follows: 

\begin{enumerate}
	\item
		We introduce a weighted variant of \LStar\ parametric on an arbitrary semiring, together with sufficient conditions for termination (Section~\ref{sec:algo}).
	\item
		We show that for general semirings our algorithm might not terminate. In particular, if the semiring is the natural numbers, one of the steps of the algorithm might not converge (Section~\ref{sec:naturals}).
	\item
		We prove that the algorithm terminates if the semiring is a {\em principal ideal domain}, covering the known case 
		of fields, but also the integers. This yields the first active learning algorithm for WFAs over the integers (Section~\ref{sec:PIDs}).
\end{enumerate}

We start in \Cref{sec:overview} by explaining the learning algorithm for WFAs over the reals and pointing out the challenges in extending it to arbitrary semirings.

\section{Overview of the Approach}\label{sec:overview}

In this section, we give an overview of the work developed in the paper through examples. We start by informally explaining the general
algorithm for learning weighted automata that we introduce in Section~\ref{sec:algo}, for the case
where the semiring is a field. More specifically, for simplicity we consider the field of real numbers throughout this section. Later in the section, we illustrate why this algorithm does not work for an arbitrary semiring. 

Angluin's \LStar\ algorithm provides a procedure to learn the minimal DFA accepting a certain (unknown) regular language.
In the weighted variant we will introduce in Section~\ref{sec:algo}, for the specific case of the field of real numbers, the algorithm produces the minimal WFA accepting a weighted rational language (or formal power series) $\lang \colon A^* \to \R$.

A WFA over $\R$ consists of a set of states, a linear combination of initial states, a transition function that for each state and input symbol produces a linear combination of successor states, and an output value in $\R$ for each state (Definition~\ref{def:wfa}).
As an example, consider the WFA over $A = \{a\}$ below. 
\[
	\begin{automaton}[baseline=-.5ex]
		\node[initial,state] (q0) {$q_0/2$};
		\node[state,right of=q0] (q1) {$q_1/3$};
		\path (q0) edge node[trlab,above]{$a,1$} (q1)
		(q0) edge[loop above] node[trlab,above]{$a,1$} (q0)
		(q1) edge[loop above] node[trlab,above]{$a,2$} (q1);
	\end{automaton}
\]
Here $q_0$ is the only initial state, with weight 1, as indicated by the arrow into it that has no origin.
When reading $a$, $q_0$ transitions with weight 1 to itself and also with weight 1 to $q_1$; $q_1$ transitions with weight 2 just to itself.
The output of $q_0$ is $2$ and the output of $q_1$ is $3$.

The language of a WFA is determined by letting it read a given word and determining the final output according to the weights and outputs assigned to individual states.
More precisely, suppose we want to read the word $aaa$ in the example WFA above.
Initially, $q_0$ is assigned weight 1 and $q_1$ weight 0.
Processing the first $a$ then leads to $q_0$ retaining weight 1, as it has a self-loop with weight 1, and $q_1$ obtaining weight 1 as well.
With the next $a$, the weight of $q_0$ still remains $1$, but the weight of $q_1$ doubles due to its self-loop of weight 1 and is added to the weight 1 coming from $q_0$, leading to a total of $3$.
Similarly, after the last $a$ the weights are $1$ for $q_0$ and $7$ for $q_1$.
Since $q_0$ has output $2$ and $q_1$ output $3$, the final result is $2 \cdot 1 + 3 \cdot 7 = 23$.

The learning algorithm assumes access to a \emph{teacher} (sometimes also called \emph{oracle}), who answers two types of queries:
\begin{itemize}
	\item
		\emph{membership queries}, consisting of a single word $w \in A^*$, to which the teacher replies with a weight $\lang(w) \in \R$;
	\item
		\emph{equivalence queries}, consisting of a hypothesis WFA $\aut$, to which the teacher replies \textbf{yes} if its language $\lang_\aut$ equals the target language $\lang$ and \textbf{no} otherwise, providing a counterexample $w \in A^*$ such that $\lang(w) \neq \lang_\aut(w)$.
\end{itemize}
In practice, membership queries are often easily implemented by interacting with the system one wants to model the behaviour of.
However, equivalence queries are more complicated---as the perfect teacher does not exist and the target automaton is not known they are commonly approximated by testing.
Such testing can however be done exhaustively if a bound on the number of states of the target automaton is known.
Equivalence queries can also be implemented exactly when learning algorithms are being compared experimentally on generated automata whose languages form the targets.
In this case, standard methods for language equivalence, such as the ones based on bisimulations~\cite{boreale2009}, can be used.

The learning algorithm incrementally builds an {\em observation table}, which at each stage contains partial information about the language $\lang$ determined by two finite sets $S, E \subseteq A^*$.
The algorithm fills the table through membership queries.
As an example, and to set notation, consider the following table (over $A=\{a\}$).
%\begin{center}%\vspace{-.8cm}
	\begin{tabular}{m{.4\linewidth}m{.5\linewidth}}
		\centering
		\vspace{-1.0cm} % G: not sure why I needed this...
		\begin{tabular}{r r c c c}
			& & \multicolumn{3}{c}{$\overbracket[.8pt][2pt]{\rule{7ex}{0pt}}^{\displaystyle E}$} \\[-.05cm]
			& & \multicolumn{1}{|c}{$\eword$} & $a$ & $aa$ \\
			\cline{2-5}
			\ldelim[{2}{2em}[$\begin{array}{@{}c@{}}S\\ \vspace{-.45cm}\mbox{}\end{array}$] & $\eword$ & \multicolumn{1}{|c}{0} & 1 & 3 \\
			\ldelim[{2}{2.9em}[$\begin{array}{@{}c@{}}\mbox{}\\S \cdot A\end{array}$]	& $a$ & \multicolumn{1}{|c}{1} & 3 & 7 \\[.04cm]
			\cline{2-5}
			 & $aa$ & \multicolumn{1}{|c}{3} & 7 & 15
		\end{tabular} &
		{
			\begin{gather*}
				\row \colon S \to \R^E \\
				\row(u)(v) = \lang(uv) \\[1ex]
				\srow \colon S \cdot A \to \R^E \\
				\srow(ua)(v) = \lang(uav)
			\end{gather*}
		}
	\end{tabular}
%\end{center}
%\begin{tightcenter}
%	\begin{tabular}{m{.4\linewidth}m{.5\linewidth}}
%		\centering
%		\begin{tabular}{L{4em} L{1ex} | C{.5ex}C{.5ex}C{.5ex}@{}m{0pt}@{}}
%			\multicolumn{5}{c}{$\hspace{7ex}\overbracket[.8pt][2pt]{\rule{11ex}{0pt}}^{\displaystyle E}$} \\
%			& & \eword &  a & aa 
%			\\
%			\cline{2-5}
%			\ldelim[{2}{2.5em}[$\begin{array}{@{}c@{}}S\\ \cup\\ S \cdot A\end{array}$] & \eword & \lmark{r1} 0 & 1 & 3 \rmark{r1}  & 
%			\\[1ex]
%			\cline{2-5}
%			& a &  \lmark{r2}  1 & 3 & 7 \rmark{r2} &  
%		\end{tabular}
%
%		\smallskip
%		$S,E \subseteq A^*$
%		&
%		{
%			\begin{gather*}
%				\invlmark{r4} \row \colon S \cup S\cdot A \to \R^E \invrmark{r4} \\[3ex]
%				\begin{aligned}
%					\row(u)(v) = \lang(uv)
%				\end{aligned}
%			\end{gather*}
%		}
%	\end{tabular}
%	%
%	\begin{tikzpicture}[overlay,remember picture,>=stealth']
%		\path (r1.east) edge node[trlab,left] {} (r4.west);		
%		\path (r2.east) edge[-{Stealth[length=5pt,width=8pt]}] node[trlab,right] {} (r4.west);	
%	\end{tikzpicture}
%\end{tightcenter}
%

This table indicates that $\lang$ assigns $0$ to $\eword$, $1$ to $a$, $3$ to $aa$, $7$ to $aaa$, and $15$ to $aaaa$.
For instance, we see that $\row(a)(aa) = \srow(aa)(a) = 7$.
Since $\row$ and $\srow$ are fully determined by the language $\lang$, we will refer to an observation table as a pair $(S,E)$, leaving the language $\lang$ implicit.

If the observation table $(S,E)$ satisfies certain properties described below, then it represents a WFA $(S,\delta,i,o)$, called the \emph{hypothesis}, as follows:
\begin{itemize}
	\item
		$\delta \colon S \to (\R^S)^A$ is a linear map defined by choosing for $\delta(s)(a)$ a linear combination over $S$ of which the rows evaluate to $\srow(sa)$;
	\item
		$i \colon S \to \R$ is the initial weight map defined as $i(\eword) = 1$ and $i(s) = 0$ for $s \ne \eword$;
	\item
		$o \colon S \to \R$ is the output weight map defined as $o(s) = \row(s)(\eword)$.
\end{itemize}
For this to be well-defined, we need to have $\eword \in S$ (for the initial weights) and $\eword \in E$ (for the output weights), and for the transition function there is a crucial property of the table that needs to hold: closedness.
In the weighted setting, a table is closed if for all $t \in S \cdot A$, there exist $r_s \in \R$ for all $s \in S$ such that
\[
	\srow(t) = \sum_{s \in S} r_s \cdot \row(s).
\]
If this is not the case for a given $t \in S \cdot A$, the algorithm adds $t$ to $S$.
The table is repeatedly extended in this manner until it is closed.
The algorithm then constructs a hypothesis, using the closedness witnesses to determine transitions, and poses an equivalence query to the teacher.
It terminates when the answer is \textbf{yes}; otherwise it extends the table with the counterexample provided by adding all its suffixes to $E$,
and the procedure continues by closing again the resulting table. 
In the next subsection we describe
the algorithm through an example. 

%Here the semimodule structure on the rows is defined pointwise.
%Each time the algorithm constructs an automaton, it poses an equivalence query to the teacher.
%It terminates when the answer is \textbf{yes}; otherwise it extends the table with the counterexample provided.

\begin{remark}
	The original \LStar\ algorithm requires a second property to construct a hypothesis, called \emph{consistency}.
	Consistency is difficult to check in extended settings, so the present paper is based on a variant of the algorithm inspired by Maler and Pnueli~\cite{maler1995} where only closedness is checked and counterexamples are handled differently.
%	\jurriaan{Isn't this also  inspired by~\cite{BolligHKL09}?}\gerco{That's a later paper. It's less inspired by that, because they still do a bit of consistency. We will mention it at the NFA example.}
	See~\cite{vanheerdt2017} for an overview of consistency in different settings.
\end{remark}

\subsection{Example: Learning a Weighted Language over the Reals}\label{sec:execution_example_original}

Throughout this section we consider the following weighted language:
	\begin{align*}
		\lang \colon \{a\}^* \to \R &
			&
			\lang(a^j) = 2^j - 1.
	\end{align*}
The minimal WFA recognising it has 2 states.
We will illustrate how the weighted variant of Angluin's algorithm recovers this WFA.

%\tightpar{\algstep{1}.}
We start from $S=E=\{\eword\}$, and fill the entries of the table on the left below by asking membership queries for $\eword$ and $a$.
The table is not closed and hence we build the table on its right, adding the membership result for $aa$.
The resulting table is closed, as $\srow(aa) = 3 \cdot \row(a)$, so we construct the hypothesis $\aut_1$.
\begin{center}
%\vspace{-.8cm}
	\begin{tabular}{m{.4\linewidth}m{.4\linewidth}}
		\vspace{-.5cm}
		\begin{tabular}{l | c}
			& $\eword$ \\
			\hline
			$\eword$ & 0 \\
			\hline
			$a$ & 1
		\end{tabular}
		\qquad\qquad\qquad 
		\begin{tabular}{r | c}
			& $\eword$ \\
			\hline
			$\eword$ & 0 \\
			$a$ & 1 \\
			\hline
			$aa$ & 3
		\end{tabular}
		&
		\begin{gather*}
			\aut_1 =
			\begin{automaton}[baseline=-.5ex]
				\node[initial,state] (q0) {$q_0/0$};
				\node[state,right of=q0] (q1)  {$q_1/1$};
				\path (q0) edge node[trlab,above]{$a,1$} (q1)
				(q1) edge[loop right] node[trlab,right]{$a,3$} (q1); 	
			\end{automaton}	\\
			q_0 = \eword \\
			q_1 = a
		\end{gather*}
	\end{tabular}
\end{center}
The teacher replies \textbf{no} and gives the counterexample $aaa$, which is assigned $9$ by the hypothesis automaton $\aut_1$ but $7$ in the language.
Therefore, we extend $E \gets E \cup \{a,aa,aaa\}$.
%
%\tightpar{\algstep{2}.}
\label{alg:step2}
The table becomes the one below. It is closed, as $\srow(aa) = 3 \cdot \row(a) - 2 \cdot \row(\eword)$, so we construct a new hypothesis $\aut_2$.
\begin{center}
	\begin{tabular}{m{.4\linewidth}m{.4\linewidth}}
		\centering
		\begin{tabular}{r | c c c c}
			& $\eword$ & $a$ & $aa$ & $aaa$ \\
			\hline
			$\eword$ & 0 & 1 & 3 & 7 \\
			$a$ & 1 & 3 & 7 & 15 \\
			\hline
			$aa$ & 3 & 7 & 15 & 31
		\end{tabular} &
		\centering
		\begin{gather*}
			\aut_2 = 
			\hspace{-10pt}
			\begin{gathered}
				\begin{automaton}
					\node[initial,state] (q0) {$q_0/0$};	
					\node[state,right of=q0] (q1) {$q_1/1$};	
					\path 
					(q1) edge[loop right] node[trlab,right] {$a,3$} (q1)
					(q0) edge[bend right] node[trlab,below] {$a,1$} (q1)
					(q1) edge[bend right] node[trlab,above] {$a,-2$} (q0);
				\end{automaton}
			\end{gathered}	
		\end{gather*}
	\end{tabular}
\end{center}
The teacher replies {\bf yes} because $\aut_2$ accepts the intended language assigning $2^{j}-1 \in \R$ to the word $a^j$,
and the algorithm terminates with the correct automaton.

\subsection{Learning Weighted Languages over Arbitrary Semirings}\label{sec:overview-arbitrary}
Consider now the same language as above, but represented as a map over the semiring of natural numbers $\lang \colon \{a\}^* \to \N$
instead of a map $\lang \colon \{a\}^* \to \R$ over the reals. Accordingly, we consider a variant
of the learning algorithm over the semiring $\N$ rather than the algorithm over $\R$ described above. 
For the first part, the run of the algorithm for $\N$ is the same as above, but after receiving the counterexample we can no longer observe that $\srow(aa) = 3 \cdot \row(a) - 2 \cdot \row(\eword)$, since $-2 \not\in \N$.
In fact, there are no $m, n \in \N$ such that $\srow(aa) = m \cdot \row(\eword) + n \cdot \row(a)$.
To see this, consider the first two columns in the table and note that $\frac{3}{7}$ is bigger than $\frac{0}{1} = 0$ and $\frac{1}{3}$, so it cannot be obtained as a linear combination of the latter two using natural numbers.
We thus have a closedness defect and update $S \gets S \cup \{aa\}$, leading to the table below.
\begin{center}
	\begin{tabular}{r | c c c c}
		& $\eword$ & $a$ & $aa$ & $aaa$ \\
		\hline
		$\eword$ & 0 & 1 & 3 & 7 \\
		$a$ & 1 & 3 & 7 & 15 \\
		$aa$ & 3 & 7 & 15 & 31 \\
		\hline
		$aaa$ & 7 & 15 & 31 & 63
	\end{tabular}
\end{center}
Again, the table is not closed, since $\frac{7}{15} > \frac{3}{7}$. In fact, these closedness defects continue appearing indefinitely, leading to non-termination of the algorithm.
This is shown formally in Section~\ref{sec:naturals}.

Note, however, that there does exist a WFA over $\N$ accepting this language:
\begin{equation}\label{eq:famous-automaton-over-N}
	\begin{automaton}[baseline=-.5ex]
		\node[initial,state] (q0) {$q_0/0$};
		\node[state,right of=q0] (q1)  {$q_1/1$};
		\path (q0) edge node[trlab,above]{$a,1$} (q1)
		(q0) edge[loop above] node[trlab,above]{$a,1$} (q0)
		(q1) edge[loop above] node[trlab,above]{$a,2$} (q1); 	
	\end{automaton}
\end{equation}

The reason that the algorithm cannot find the correct automaton is closely related to the algebraic structure induced by the semiring.
In the case of the reals, the algebras are vector spaces and the closedness checks induce increases in the dimension of the hypothesis WFA, which in turn cannot exceed the dimension of the minimal one for the language.
In the case of commutative monoids, the algebras for the natural numbers, the notion of dimension does not exist and unfortunately the algorithm does not terminate.
In Section~\ref{sec:PIDs} we show that one can get around this problem for a class of semirings which includes the integers.

We mentioned earlier that during experimental evaluation the target WFA is known, and equivalence queries may be implemented via standard language equivalence methods.
A further issue with arbitrary semirings is that language equivalence can be undecidable; that is the case, e.g., for the tropical semiring.

In Section~\ref{sec:prelims} we recall basic definitions used throughout the paper, after which Section~\ref{sec:algo} introduces our general algorithm with its (parameterised) termination proof of Theorem~\ref{thm:termination}.
We then proceed to prove non-termination of the example discussed above over the natural numbers in Section~\ref{sec:naturals} before instantiating our algorithm to PIDs in Section~\ref{sec:PIDs} and showing that it terminates in Theorem~\ref{thm:pids}.
We conclude with a discussion of related and future work in Section~\ref{sec:discussion}.

\section{Preliminaries}\label{sec:prelims}

Throughout this paper we fix a semiring\footnote{%
	Rings and semirings considered in this paper are taken to be unital.
} $\Semi$ and a finite alphabet $A$.
We start with basic definitions related to semimodules and  weighted languages. 

\begin{definition}[Semimodule]
	A \emph{(left) semimodule} $M$ over $\Semi$ consists of a monoid structure on $M$, written using $+$ as the operation and $0$ as the unit, together with a scalar multiplication map ${\cdot} \colon \Semi \times M \to M$ such that:
	\begin{align*}
		s \cdot 0_M &
			= 0_M &
			0_\Semi \cdot m &
			= 0_M &
			1 \cdot m &
			= m \\
		s \cdot (m + n) &
			= s \cdot m + s \cdot n &
			(s + r) \cdot m &
			= s \cdot m + r \cdot m &
			(sr) \cdot m &
			= s \cdot (r \cdot m).
	\end{align*}
	When the semiring is in fact a ring, we speak of a \emph{module} rather than a semimodule.
	In the case of a field, the concept instantiates to a vector space.
\end{definition}

As an example, commutative monoids are the semimodules over the semiring of natural numbers.
Any semiring forms a semimodule over itself by instantiating the scalar multiplication map to the internal multiplication.
If $X$ is any set and $M$ is a semimodule, then $M^X$ with pointwise operations also forms a semimodule.
A similar semimodule is the \emph{free semimodule} over $X$, which differs from $M^X$ in that it fixes $M$ to be $\Semi$ and requires its elements to have \emph{finite support}.
This enables an important operation called \emph{linearisation}.

\begin{definition}[Free semimodule]
	The \emph{free semimodule} over a set $X$ is given by the set
	\[
		\free(X) = \{f \colon X \to \Semi \mid \supp(f)\text{ is finite}\}
	\]
	with pointwise operations.
	Here $\supp(f) = \{x \in X \mid f(x) \ne 0\}$.
	We sometimes identify the elements of $\free(X)$ with formal sums over $X$.
	Any semimodule isomorphic to $\free(X)$ for some set $X$ is called free.
\end{definition}
If $X$ is a finite set, then $\free(X) = \Semi^X$.
We now define {\em linearisation} of a function into a semimodule, which uniquely extends it to a semimodule homomorphism, witnessing the fact that $\free(X)$ is free.

\begin{definition}[Linearisation]\label{def:linearisation}
	Given a set $X$, a semimodule $M$, and a function $f \colon X \to M$, we define the \emph{linearisation} of $f$ as the semimodule homomorphism $f^\sharp \colon \free(X) \to M$ given by
	\[
		f^\sharp(\alpha) = \sum_{x \in X} \alpha(x) \cdot f(x).
	\]
	The $(-)^\sharp$ operation has an inverse that maps a semimodule homomorphism $g \colon \free(X) \to M$ to the function $g^\dagger \colon X \to M$ given by
	\[
		g^\dagger(x) = g(\partial_x), \qquad\qquad \partial_x(y) = \begin{cases}
			1 &
				\text{if $y = x$} \\
			0 &
				\text{if $y \ne x$}.
		\end{cases}
	\]
\end{definition}

We proceed with the definition of WFAs and their languages.

\begin{definition}[WFA]\label{def:wfa}
	A \emph{weighted finite automaton (WFA)} over $\Semi$ is a tuple $(Q, \delta, i, o)$, where $Q$ is a finite set, $\delta \colon Q \to (\Semi^Q)^A$, and $i, o \colon Q \to \Semi$.
\end{definition}
A \emph{weighted language} (or just \emph{language}) over $\Semi$ is a function $A^* \to \Semi$.
To define the language accepted by a WFA $\aut = (Q, \delta, i, o)$,
we first introduce the notions of \emph{observability map} $\obs_\aut \colon \free(Q) \to \Semi^{A^*}$
and \emph{reachability map} $\reach_\aut \colon \free(A^*) \to \free(Q)$
as the semimodule homomorphisms given by
	\begin{align*}
		\reach_\aut^\dagger(\eword) &
			= i &
			\obs_\aut(m)(\eword) &
			= o^\sharp(m) \\
		\reach_\aut^\dagger(ua) &
			= \delta^\sharp(\reach_\aut^\dagger(u))(a) &
			\obs_\aut(m)(au) &
			= \obs_\aut(\delta^\sharp(m)(a))(u).
	\end{align*}
%
%To define the language accepted by a weighted automaton, we first introduce the notion of \emph{observability map}, 
%and its dual, the reachability map. 
%
%\begin{definition}[Reachability and observability maps]
%	Given a WFA $\aut = (Q, \delta, i, o)$, its \emph{reachability map} $\reach_\aut \colon \free(A^*) \to \free(Q)$ and \emph{observability map} $\obs_\aut \colon \free(Q) \to \Semi^{A^*}$ are the semimodule homomorphisms given by
%	\begin{align*}
%		\reach_\aut^\dagger(\eword) &
%			= i &
%			\obs_\aut(m)(\eword) &
%			= o^\sharp(m) \\
%		\reach_\aut^\dagger(ua) &
%			= \delta^\sharp(\reach_\aut^\dagger(u))(a) &
%			\obs_\aut(m)(au) &
%			= \obs_\aut(\delta^\sharp(m)(a))(u).
%	\end{align*}
%\end{definition}
%
%\begin{definition}[Weighted language]
%\end{definition}
The \emph{language accepted by a WFA} $\aut = (Q, \delta, i, o)$ is the function $\lang_\aut \colon A^* \to \Semi$ given by $\lang_\aut = \obs_\aut(i)$.
Equivalently, one can define this as $\lang_\aut = o^\sharp \circ \reach_\aut^\dagger$.

\section{General Algorithm for WFAs}\label{sec:algo}

In this section we define the general algorithm for WFAs over $\Semi$, 
as described informally in Section~\ref{sec:overview}. Our algorithm assumes 
the existence of a \emph{closedness strategy} (Definition~\ref{def:closedness-strategy}),
which allows one to check whether a table is closed, and in case it is, provide relevant witnesses. 
We then introduce sufficient conditions on $\Semi$ and on the language $\lang$ to be learned 
under which the algorithm terminates. 
	
\begin{definition}[Observation table]
	An \emph{observation table} (or just \emph{table}) $(S, E)$ consists of two sets $S, E \subseteq A^*$.
	We write $\tables = \Psf(A^*) \times \Psf(A^*)$ for the set of finite tables (where $\Psf$(X) denotes the collection of finite subsets of a set $X$).
	Given a language $\lang \colon A^* \to \Semi$, an observation table $(S, E)$ determines the \emph{row function} $\row_{(S, E, \lang)} \colon S \to \Semi^E$ and the \emph{successor row} function $\srow_{(S, E, \lang)} \colon S \cdot A \to \Semi^E$ as follows:
	\begin{align*}
		\row_{(S, E, \lang)}(w)(v) &
			= \lang(wv) &
			\srow_{(S, E, \lang)}(wa)(v) &
			= \lang(wav).
	\end{align*}
	We often write $\row_\lang$ and $\srow_\lang$, or even $\row$ and $\srow$, when the parameters are clear from the context.
\end{definition}

A table is \emph{closed} if the successor rows are linear combinations of the existing rows in $S$. 
To make this precise, we use the linearisation $\row^\sharp$ (Definition~\ref{def:linearisation}),
which extends $\row$ to linear combinations of words in $S$.

\begin{definition}[Closedness]
	Given a language $\lang$, a table $(S, E)$ is \emph{closed} if for all $w \in S$ and $a \in A$ there exists $\alpha \in \free(S)$ such that $\srow(wa) = \row^\sharp(\alpha)$.
\end{definition} 
This corresponds to the notion of closedness described in Section~\ref{sec:overview}.

A further important ingredient of the algorithm is a method for checking whether a table is closed.
This is captured by the notion of closedness strategy.

\begin{definition}[Closedness strategy]\label{def:closedness-strategy}
	Given a language $\lang$, a \emph{closedness strategy} for $\lang$ is a family of computable functions $$\left(\cs_{(S, E)} \colon S \cdot A \to \{\bot\} \cup \free(S)\right)_{(S,E) \in \tables}$$ satisfying the following two properties:
	\begin{itemize}
		\item
			if $\cs_{(S, E)}(t) = \bot$, then there is no $\alpha \in \free(S)$ s.t.\ $\row^\sharp(\alpha) = \srow(t)$, and
		\item
			if $\cs_{(S, E)}(t) \ne \bot$, then $\row^\sharp(\cs_{(S, E)}(t)) = \srow(t)$.
	\end{itemize}
\end{definition}
Thus, given a closedness strategy as above, a table $(S,E)$ is closed iff
$\cs_{(S, E)}(t) \ne \bot$ for all $t \in S \cdot A$. More specifically, 
for each $t \in S \cdot A$ we have 
that $\cs_{(S, E)}(t) \ne \bot$ iff 
the (successor) row corresponding to $t$ already forms a linear combination of rows labelled by $S$.
In that case, this linear combination is returned by $\cs_{(S, E)}(t)$. This is used
to close tables in our learning algorithm, introduced below.

Examples of semirings and (classes of) languages that admit a closedness strategy are described at the end of this section. 
Important for our algorithm will be that closedness strategies are computable. This problem is equivalent 
to solving systems of equations $A \underline{x} = \underline{b}$,
where $A$ is the matrix whose columns are $\row(s)$ for $s \in S$, $\underline{x}$ is a vector of length $|S|$, and 
$\underline{b}$ is the vector consisting of the row entries in $\srow(t)$ for some $t \in S \cdot A$. 
These observations motivate the following definition.

\begin{definition}[Solvability]\label{def:solvability}
	A semiring $\Semi$ is \emph{solvable} if a solution to any finite system of linear equations of the form $A \underline{x} = \underline{b}$ is computable.
\end{definition}

We have the following correspondence.

\begin{proposition}\label{prop:strategy}
	For any language accepted by a WFA over any semiring there exists a closedness strategy if and only if the semiring is solvable.
\end{proposition}
\begin{proof}
	If the semiring is solvable, we obtain a closedness strategy by the remarks prior to \Cref{def:solvability}.
	Conversely, we can construct a language that is non-zero on finitely many words and encode in a table $(S, E)$ a given linear equation.
	To be able to freely choose the value in each table cell, we can consider a sufficiently large alphabet to make sure $S$ and $E$ contain only single-letter words.
	This avoids dependencies within the table.
	\qed
\end{proof}

\begin{algorithm}[H]
	\caption{Abstract learning algorithm for WFA over $\Semi$}\label{alg:main}
	\begin{algorithmic}[1]
		\State $S, E \gets \{\eword\}$
		\While{\texttt{true}} 
			\While {$\cs_{(S, E)}(t) = \bot$ for some $t \in S \cdot A$} \label{alg:while1}
				\State $S \gets S \cup \{t\}$ \label{alg:while2}
			\EndWhile
			\For{$s \in S$} \label{alg:hyp-begin}
				\State $o(s) \gets \row_\lang(s)(\eword)$
				\For{$a \in A$}
					\State $\delta(s)(a) \gets \cs_{(S,E)}(sa)$\label{alg:hyp-end}
				\EndFor
			\EndFor
			\If{$\mathsf{EQ}(S, \delta, \eword, o) = w \in A^*$}\label{alg:eq-check}
				\State $E \gets E \cup \suffixes(w)$
			\Else
				\State \textbf{return} $(S, \delta, \eword, o)$\label{alg:ret}
			\EndIf
		\EndWhile
	\end{algorithmic}
\end{algorithm}

We now have all the ingredients to formulate the algorithm to learn weighted languages over a general semiring. The pseudocode is displayed in Algorithm~\ref{alg:main}.

The algorithm keeps a table $(S,E)$, and starts by initialising both $S$ and $E$ to contain just the empty word.
The inner while loop (lines~\ref{alg:while1}--\ref{alg:while2}) uses the closedness strategy to repeatedly 
check whether the current table is closed and add new rows in case it is not. Once the table is closed, 
a hypothesis is constructed, again using the closedness strategy (lines~\ref{alg:hyp-begin}--\ref{alg:hyp-end}). 
This hypothesis $(S,\delta, \eword, o)$ is then given to the teacher for an equivalence check.
The equivalence check is modelled by $\mathsf{EQ}$ (line~\ref{alg:eq-check}) as follows: 
 if the hypothesis is incorrect,
the teacher non-deterministically returns a counterexample $w \in A^*$, the condition evaluates to \texttt{true}, 
and the suffixes of $w$ are added to $E$;
otherwise, if the hypothesis is correct, the condition on line~\ref{alg:eq-check} evaluates to
\texttt{false}, and the algorithm returns the correct hypothesis on line~\ref{alg:ret}. 

\subsection{Termination of the General Algorithm}

The main question remaining is: under which conditions does this algorithm terminate and hence learns the unknown weighted language? We proceed to give abstract conditions under which it terminates. 
There are two main assumptions: 
\begin{enumerate}
	\item
		A way of measuring progress the algorithm makes with the observation table when it distinguishes linear combinations of rows that were previously equal, together with a bound on this progress (Definition~\ref{def:progress-measure}).
	\item
		An assumption on the \emph{Hankel matrix} of the input language (Definition~\ref{def:hankel}), which makes sure we encounter finitely many closedness defects throughout any run of the algorithm. More specifically, we assume that the Hankel matrix satisfies a finite approximation property (Definition~\ref{def:ascending}). %\alex{we need a forward ref here}
\end{enumerate}
The first assumption is captured by the definition of progress measure:
\begin{definition}[Progress measure]\label{def:progress-measure}
	A \emph{progress measure} for a language $\lang$ is a function $\size \colon \tables \to \N$ such that
	\begin{itemize}
		\item[(a)]
			there exists $n \in \N$ such for all $(S, E) \in \tables$ we have $\size(S, E) \le n$;
		\item[(b)]
			given $(S, E), (S, E') \in \tables$ and $s_1, s_2 \in \free(S)$ such that $E \subseteq E'$ and $\row_{(S, E, \lang)}^\sharp(s_1) = \row_{(S, E, \lang)}^\sharp(s_2)$ but $\row_{(S, E', \lang)}^\sharp(s_1) \ne \row_{(S, E', \lang)}^\sharp(s_2)$, we have $\size(S, E') > \size(S, E)$.
	\end{itemize}
\end{definition}
A progress measure assigns a `size' to each table, in such a way that (a) there is a global bound on the 
size of tables, and (b) if we extend a table with some proper tests in $E$, i.e., such that some combinations of
rows in $\row^\sharp$ that
were equal before get distinguished by a newly added test, then the size of the extended table is properly above
the size of the original table. This is used to ensure that, when adding certain counterexamples supplied by the teacher, 
the size of the table, measured according to the above $\size$ function, properly increases.
%As such, it forms the basis for a ranking function for the outside while loop of Algorithm~\ref{alg:main}.
%G: ^ commented this out because some counterexamples lead to a closedness defect instead, hence also the added word `certain'

The second assumption that we use for termination is phrased in terms of the Hankel matrix
associated to the input language $\lang$, which represents $\lang$
as the (semimodule generated by the) infinite table where both the rows and columns contain all words. 
The Hankel matrix is defined as follows.

\begin{definition}[Hankel matrix]\label{def:hankel}
	Given a language $\lang \colon A^* \to \Semi$, the \emph{semimodule generated by a table} $(S, E)$ is given by the image of $\row^\sharp$.
	We refer to the semimodule generated by the table $(A^*, A^*)$ as the \emph{Hankel matrix} of $\lang$.
\end{definition}
The Hankel matrix is approximated by the tables 
that occur during the execution of the algorithm. For termination, we will therefore assume
that this matrix satisfies the following finite approximation condition. 

\begin{definition}[Ascending chain condition]\label{def:ascending}
	We say that a semimodule $M$ satisfies the \emph{ascending chain condition} if for all inclusion chains of subsemimodules of $M$,
	\[
		S_1 \subseteq S_2 \subseteq S_3 \subseteq \cdots,
	\]
	there exists $n \in \N$ such that for all $m \ge n$ we have $S_m = S_n$.
\end{definition}

Given the notions of progress measure, Hankel matrix and ascending chain condition, we 
can formulate the general theorem for termination of Algorithm~\ref{alg:main}.  

\begin{theorem}[Termination of the abstract learning algorithm]\label{thm:termination}
	In the presence of a progress measure, Algorithm~\ref{alg:main} terminates whenever the Hankel matrix of
	the target language satisfies the ascending chain condition (\Cref{def:ascending}).
\end{theorem}
\begin{proof}
	Suppose the algorithm does not terminate. Then there is a sequence $\{(S_n, E_n)\}_{n \in \N}$ of tables where $(S_0, E_0)$ is the initial table and $(S_{n + 1}, E_{n + 1})$ is formed from $(S_n, E_n)$ after resolving a closedness defect or adding columns due to a counterexample.

	We write $H_n$ for the semimodule generated by the table $(S_n, A^*)$.
	We have $S_n \subseteq S_{n + 1}$ and thus $H_n \subseteq H_{n + 1}$.
	Note that a closedness defect for $(S_n, E_n)$ is also a closedness defect for $(S_n, A^*)$, so if we resolve the defect in the next step, the inclusion $H_n \subseteq H_{n + 1}$ is strict.
	Since these are all included in the Hankel matrix, which satisfies the ascending chain condition, there must be an $n$ such that for all $k \ge n$ we have that $(S_k, E_k)$ is closed.

	In~\cite[Section~6]{vanheerdt2017} it is shown that in a general table used for learning automata with side-effects given by a monad there exists a suffix of each counterexample for the corresponding hypothesis that when added as a column label leads to either a closedness defect or to distinguishing two combinations of rows in the table.
	Since WFAs are automata with side-effects given by the free semimodule monad\footnote{%
		We note that~\cite{vanheerdt2017} assumes the monad to preserve finite sets.
		However, the relevant arguments do not depend on this.
	} and we add all suffixes of the counterexample to the set of column labels, this also happens in our algorithm.
	Thus, for all $k \ge n$ where we process a counterexample, there must be two linear combinations of rows distinguished, as closedness is already guaranteed.
	Then the semimodule generated by $(S_k, E_k)$ is a strict quotient of the semimodule generated by $(S_{k + 1}, E_{k + 1})$.
	By the progress measure we then find $\size(S_k, E_k) < \size(S_{k + 1}, E_{k + 1})$, which cannot happen infinitely often.
	We conclude that the algorithm must terminate.
	\qed
\end{proof}

%We briefly discuss applications to two classes of semirings for which learning algorithms are already known in the literature~\cite{Bergadano,vanheerdt2017}.

To illustrate the hypotheses needed for Algorithm~\ref{alg:main} and its termination (Theorem~\ref{thm:termination}),
we consider two classes of semirings for which learning algorithms are already known in the literature~\cite{Bergadano,vanheerdt2017}.

\begin{example}[Weighted languages over fields]
	Consider any field for which the basic operations are computable.
	Solvability is then satisfied via a procedure such as Gaussian elimination, so by Proposition~\ref{prop:strategy} there exists a closedness strategy. Hence, 
	we can instantiate Algorithm~\ref{alg:main} with $\Semi$ being such a field. 
	
	For termination, we show that the hypotheses of Theorem~\ref{thm:termination} are satisfied
	whenever the input language is accepted by a WFA.
	First, a progress measure is given by the dimension of the vector space generated by the table.
	To see this, note that if we distinguish two linear combinations of rows, we can assume without loss of generality that one of these linear combinations in the extended table uses only basis elements.
	This in turn can be rewritten to distinguishing a single row from a linear combination of rows using field operations, with the property that the extended version of the single row is a basis element.
	Hence, the row was not a basis element in the original table, and therefore the dimension of the vector space generated by the table has increased.
	Adding rows and columns cannot decrease this dimension, so it is bounded by the dimension of the Hankel matrix.
	Since the language we want to learn is accepted by a WFA, the associated Hankel matrix has a finite dimension~\cite{CarlyleP71,fliess1974} (see also, e.g.,~\cite{DBLP:conf/nips/BalleM12}), providing a bound for our progress measure.

%	Since the language we want to learn is accepted by a WFA, it is well known that the Hankel matrix is generated by 
%	the languages accepted by the states of the minimal WFA accepting the language (cf.~e.g.~\cite{DBLP:conf/nips/BalleM12} and references therein).
%	This gives the Hankel matrix a finite dimension, which provides a bound for our progress measure.
	
	Finally, for any ascending chain of subspaces of the Hankel matrix, these subspaces are of finite dimension bounded by the dimension of the Hankel matrix. 
	The dimension increases along a strict subspace relation, so the chain converges.
\end{example}

\begin{example}[Weighted languages over finite semirings]
	Consider any finite semiring.
	Finiteness allows us to apply a brute force approach to solving systems of equations.
	This means the semiring is solvable, and hence a closedness strategy exists by Proposition~\ref{prop:strategy}.
	
	For termination, we can define a progress measure by assigning to each table the size of the image of $\row^\sharp$.
	Distinguishing two linear combinations of rows increases this measure.
	If the language we want to learn is accepted by a WFA, then the Hankel matrix contains 
	a subset of the linear combinations of the languages of its states.
	Since there are only finitely many such linear combinations, the Hankel matrix is finite, which bounds our measure.
	A finite semimodule such as the Hankel matrix in this case does not admit infinite chains of subspaces.
	We conclude by Theorem~\ref{thm:termination} that Algorithm~\ref{alg:main} terminates for the instance 
	that the semiring $\Semi$ is a finite, if the input language is accepted by a WFA over $\Semi$. 
\end{example}

For the Boolean semiring, an instance of the above finite semiring example, WFAs are non-deterministic finite automata.
The algorithm we recover by instantiating Algorithm~\ref{alg:main} to this case is close to the algorithm first described by Bollig et al.~\cite{BolligHKL09}.
The main differences are that in their case the hypothesis has a state space given by a minimally generating subset of the distinct rows in the table rather than all elements of $S$, and they do apply a notion of consistency.

In Section~\ref{sec:PIDs} we will show that Algorithm~\ref{alg:main} can learn WFAs over principal ideal domains---notably including the integers---thus providing a strict generalisation of existing techniques.

\section{Issues with Arbitrary Semirings}\label{sec:naturals}

We concluded the previous section with examples of semirings for which Algorithm~\ref{alg:main} terminates if the target language is accepted by a WFA. 
In this section, we prove a negative result for the algorithm over the semiring $\N$: we show that it does not terminate
on a certain language over $\N$ accepted by a WFA over $\N$, as anticipated in Section~\ref{sec:overview-arbitrary}. 
This means that Algorithm~\ref{alg:main} does not work well for arbitrary semirings.
The problem is that the Hankel matrix of a language recognised by WFA does not necessarily satisfy the 
ascending chain condition that is used to prove Theorem~\ref{thm:termination}.
In the example given in the proof below, the Hankel matrix is not even finitely generated.

\begin{theorem}
	There exists a WFA $\aut_\N$ over $\N$ such that Algorithm~\ref{alg:main} does not terminate when given $\lang_{\aut_\N}$ as input, regardless of the closedness strategy used.
\end{theorem}
\begin{proof}
	Let $\aut_\N$ be the automaton over the alphabet $\{a\}$ given in~\eqref{eq:famous-automaton-over-N} in Section~\ref{sec:overview-arbitrary}. Formally, 
	$\aut_\N = (Q, \delta, i, o)$, where
	\begin{align*}
		Q &
			= \{q_0, q_1\} &
			i &
			= q_0 &
			o(q_0) &
			= 0 \\
		\delta(q_0)(a) &
			= q_0 + q_1 &
			\delta(q_1)(a) &
			= 2q_1 &
			o(q_1) &
			= 1.
	\end{align*}
	As mentioned in Section~\ref{sec:overview-arbitrary}, the language $\lang \colon \{a\}^* \to \N$ accepted by $\aut_\N$ is given by $\lang(a^j) = 2^j - 1$.
	This can be shown more precisely as follows. First one shows by induction on $j$ that $\obs_{\aut_\N}(q_1)(a^j) = 2^j$ for all $j \in \N$---we leave the straightforward argument to the reader.
	Second, we show, again by induction on $j$, that $\obs_{\aut_\N}(q_0)(a^j) = 2^j - 1$.
	This implies the claim, as $\lang = \obs_{\aut_\N}(q_0)$.
	For $j = 0$ we have $\obs_{\aut_\N}(q_0)(a^j) = o(q_0) = 0 = 2^0-1$ as required.
	For the inductive step, let $j = k + 1$ and assume $\obs_{\aut_\N}(q_0)(a^k) = 2^k - 1$.
	We calculate
	\begin{align*}
		\obs_{\aut_\N}(q_0)(a^{k+1}) &
			= \obs_{\aut_\N}(q_0 + q_1)(a^k) \\
		&
			= \obs_{\aut_\N}(q_0)(a^k) + \obs_{\aut_\N}(q_1)(a^k) \\
		&
			= (2^k-1) + 2^k \\
		&
			= 2^{k + 1} - 1.
	\end{align*}
	Note that in particular the language $\lang$ is injective.

	Towards a contradiction, suppose the algorithm does terminate with table $(S, E)$.
	Let $J = \{j \in \N \mid a^j \in S\}$ and define $n = \mathsf{max}(J)$.
	Since the algorithm terminates with table $(S,E)$, the latter must be closed.
	In particular, there exist $k_j \in \N$ for all $j \in J$ such that $\sum_{j \in J} k_j \cdot \row_\lang(a^j) = \srow_\lang(a^na)$.
	We consider two cases.
	First assume $E = \{\eword\}$ and let $\aut = (Q', \delta', i', o')$ be the hypothesis.
	For all $l \in \N$ we have $\row_\lang^\sharp(\reach_\aut^\dagger(a^l))(\eword) = 2^l - 1$ because $\aut$ must be correct.
	Thus, if $a^l \in S \cdot A$, then $\row_\lang^\sharp(\reach_\aut^\dagger(a^l)) = \srow_\lang(a^l)$.
	In particular,
	\[
		\row_\lang^\sharp(\reach_\aut^\dagger(a^na)) = \srow_\lang(a^na) = \sum_{j \in J} k_j \cdot \row_\lang(a^j).
	\]
	Note that we can choose the $k_j$ such that $\reach_\aut^\dagger(a^na) = \sum_{j \in J} k_j \cdot a^j$.
	Since
	\begin{align*}
		\row_\lang^\sharp\left(\delta'^\sharp\left(\sum_{j \in J} k_j \cdot a^j\right)(a)\right) &
			= \row_\lang^\sharp\left(\sum_{j \in J} k_j \cdot \delta'(a^j)(a)\right) \\
		&
			= \sum_{j \in J} k_j \cdot \row_\lang(\delta'(a^j)(a)) \\
		&
			= \sum_{j \in J} k_j \cdot \srow_\lang(a^ja),
	\end{align*}
	we have
	$
		\row_\lang^\sharp(\reach_\aut^\dagger(a^naa)) = \sum_{j \in J} k_j \cdot \srow_\lang(a^ja)
	$
	and therefore
	\[
		\sum_{j \in J} k_j \cdot \srow_\lang(a^ja)(\eword) = \row_\lang^\sharp(\reach_\aut^\dagger(a^naa))(\eword) = 2^{n + 2} - 1.
	\]
	Then
	\begin{align*}
		2^{n + 2} - 1 &
			= \sum_{j \in J} k_j \cdot \srow_\lang(a^ja)(\eword) 
			= \sum_{j \in J} k_j(2^{j + 1} - 1) \\
		&
			= 2\left(\sum_{j \in J} k_j(2^j - 1)\right) + \sum_{j \in J} k_j 
			= 2(2^{n + 1} - 1) + \sum_{j \in J} k_j,
	\end{align*}
	so $\sum_{j \in J} k_j = 1$.
	This is only possible if there is $j_1 \in J$ s.t.\ $k_{j_1} = 1$ and $k_j = 0$ for all $j \in J \setminus \{j_1\}$.
	However, this implies that $\row_\lang(a^{j_1}) = \srow_\lang(a^na)$, which contradicts injectivity of $\lang$ as $n \ge j_1$.
	%We conclude that $(S, E)$ is not closed and thus that the algorithm did not terminate.
	Thus, the algorithm did not terminate.

	For the other case, assume there is $a^m \in E$ such that $m \ge 1$.
	We have
	\[
		2^{n + 1} - 1 = \srow_\lang(a^na)(\eword) = \sum_{j \in J} k_j \cdot \row_\lang(a^j)(\eword) = \sum_{j \in J} k_j(2^j - 1),
	\]
	so
	\begin{align*}
		\sum_{j \in J} k_j(2^{j + m} - 1) &
			= \sum_{j \in J} k_j \cdot \row_\lang(a^j)(a^m) \\
		&
			= \srow_\lang(a^na)(a^m) \\
		&
			= 2^{n + m + 1} - 1 \\
		&
			= 2^m(2^{n + 1} - 1) + 2^m - 1 \\
		&
			= 2^m\left(\sum_{j \in J} k_j(2^j - 1)\right) + 2^m - 1 \\
		&
			= \left(\sum_{j \in J} k_j(2^{j + m} - 2^m)\right) + 2^m - 1 \\
		&
			= \left(\sum_{j \in J} k_j(2^{j + m} - 1)\right) + \left(\sum_{j \in J} k_j(1 - 2^m)\right) + 2^m - 1.
	\end{align*}
	Then
	\[
		\left(\sum_{j \in J} k_j(1 - 2^m)\right) + 2^m - 1 = 0.
	\]
	Since $m \ge 1$ this is only possible if there is $j_1 \in J$ s.t.\ $k_{j_1} = 1$ and $k_j = 0$ for all $j \in J \setminus \{j_1\}$.
	However, this implies $\row_\lang(a^{j_1}) = \srow_\lang(a^na)$, which again contradicts injectivity of $\lang$ as $n \ge j_1$.
	%We conclude that $(S, E)$ is not closed and thus that the algorithm did not terminate.
	Thus, the algorithm did not terminate. 
	\qed
\end{proof}

\begin{remark}
	Our proof shows non-termination for a bigger class of algorithms than Algorithm~\ref{alg:main}; it uses only the definition of the hypothesis, that closedness is satisfied before constructing the hypothesis, that $S$ and $E$ contain the empty word, and that termination implies correctness.
	For instance, adding the prefixes of a counterexample to $S$ instead of its suffixes to $E$ will not fix the issue.
\end{remark}

We have thus shown that our algorithm does not instantiate to a terminating one for an arbitrary semiring.
To contrast this negative result, in the next section we identify a class of semirings not previously explored in the learning literature where we can guarantee a terminating instantiation.

\section{Learning WFAs over PIDs}\label{sec:PIDs}

We show that for a subclass of semirings, namely {\em principal ideal domains (PIDs)}, the abstract learning algorithm of \Cref{sec:algo} terminates.
This subclass includes the integers, Gaussian integers, and rings of polynomials in one variable with coefficients in a field.
We will prove that the Hankel matrix of a language over a PID accepted by a WFA has analogous properties to those of vector spaces---finite rank, a notion of progress measure, and the ascending chain condition. We also give a sufficient condition for PIDs to be solvable, which by \Cref{prop:strategy} guarantees the existence of a closedness strategy for the learning algorithm.

To define PIDs, we first need to introduce ideals.
Given a ring $\Semi$, a \emph{(left) ideal} $I$ of $\Semi$ is an additive subgroup of $\Semi$ s.t.\ for all $s \in \Semi$ and $i \in I$ we have $si \in I$.
The ideal $I$ is \emph{(left) principal} if it is of the form $I = \Semi{}s$ for some $s \in \Semi$.

\begin{definition}[PID]
	A \emph{principal ideal domain} $\PID$ is a non-zero commutative ring in which every ideal is principal and where for all $p_1, p_2 \in \PID$ such that $p_1p_2 = 0$ we have $p_1 = 0$ or $p_2 = 0$.
\end{definition}

A module $M$ over a PID $\PID$ is called \emph{torsion free} if for all $p \in \PID$ and any $m \in M$ such that $p \cdot m = 0$ we have $p = 0$ or $m = 0$.
It is a standard result that a module over a PID is torsion free if and only if it is free~\cite[Theorem~3.10]{jacobson2012}.

The next definition of \emph{rank} is analogous to that of the dimension of a vector space and will form the basis for the progress measure.

\begin{definition}[Rank]
	We define the \emph{rank} of a finitely generated free module $\free(X)$ over a PID as $\rank(\free(X)) = |X|$.
\end{definition}

This definition extends to any finitely generated free module over a PID, as $\free(X) \cong \free(Y)$ for finite sets $X$ and $Y$ implies $|X| = |Y|$~\cite[Theorem~3.4]{jacobson2012}.

Now that we have a candidate for a progress measure function, we need to prove it has the required properties.
The following lemmas will help with this.

\begin{lemma}\label{lem:nakayama}
	Given finitely generated free modules $M,N$ over a PID s.t.\ $\rank(M) \ge \rank(N)$, any surjective module homomorphism $f \colon N \to M$ is injective.
\end{lemma}
\begin{proof}
	Since $\rank(M) \ge \rank(N)$, there exists a surjective module homomorphism $g \colon M \to N$.
	Therefore $g \circ f \colon N \to N$ is surjective and by~\cite{orzech1971} an iso.
	In particular, $f$ is injective.
	\qed
\end{proof}

\begin{lemma}\label{lem:quotrank}
	If $M$ and $N$ are finitely generated free modules over a PID such that there exists a surjective module homomorphism $f \colon N \to M$, then $\rank(M) \le \rank(N)$.
	If $f$ is not injective, then $\rank(M) < \rank(N)$.
\end{lemma}
\begin{proof}
	Let $f \colon N \to M$ be a surjective module homomorphism. Suppose towards a contradiction that 
	$\rank(M) > \rank(N)$.
	By Lemma~\ref{lem:nakayama} $f$ is injective, so $M$ is isomorphic to a submodule of $N$ and $\rank(M) \le \rank(N)$~\cite{jacobson2012}; contradiction.

	For the second part, suppose $f$ is not injective and assume towards a contradiction that $\rank(M) \ge \rank(N)$.
	Again by Lemma~\ref{lem:nakayama} $f$ is injective, which is a contradiction with our assumption.
	Thus, in this case $\rank(M) < \rank(N)$.
	\qed
\end{proof}

The lemma below states that the Hankel matrix of a weighted language over a PID has finite rank which bounds the rank of any module generated by an observation table.
This will be used to define a progress measure, used to prove termination of the learning algorithm for weighted languages over PIDs.

\begin{lemma}[Hankel matrix rank for PIDs]\label{lem:hankel}
	When targeting a language accepted by a WFA over a PID, any module generated by an observation table is free.
	Moreover, the Hankel matrix has finite rank that bounds the rank of any module generated by an observation table.
\end{lemma}
\begin{proof}
	Given a WFA $\aut = (Q, \delta, i, o)$, let $M$ be the free module generated by $Q$.
	Note that the Hankel matrix is the image of the composition $\obs_\aut \circ \reach_\aut$.
	Consider the image of the module homomorphism $\reach_\aut \colon \free(A^*) \to M$, which we write as $R$.
	Since $R$ is a submodule of $M$, we know from~\cite{jacobson2012} that $R$ is free and finitely generated with $\rank(R) \le \rank(M)$.
	The Hankel matrix can now be obtained as the image of the restriction of $\obs_\aut \colon M \to \Semi^{A^*}$ to the domain $R$.
	Let $H$ be this image, which we know is finitely generated because $R$ is.
	Since $H$ is a submodule of the torsion free module $\Semi^{A^*}$, it is also torsion free and therefore free.
	We also have a surjective module homomorphism $s \colon R \to H$, so by Lemma~\ref{lem:quotrank} we find $\rank(H) \le \rank(R)$.

	Let $N$ be the module generated by an observation table $(S, E)$.
	We have that $N$ is a quotient of the module generated by $(S, A^*)$, which in turn is a submodule of $H$.
	Using again~\cite{jacobson2012} and Lemma~\ref{lem:quotrank} we conclude that $N$ is free and finitely generated with $\rank(N) \le \rank(H)$.
	\qed
\end{proof}
The second part of Lemma~\ref{lem:hankel} would follow from a PID variant of Fliess' theorem~\cite{fliess1974}. We are not aware of such a result, and leave this for future work.  

\begin{proposition}[Progress measure for PIDs]\label{prop:pidprogress}
	There exists a progress measure for any language accepted by a WFA over a PID.
\end{proposition}
\begin{proof}
	Define $\size(S, E) = \rank(M)$, where $M$ is the module generated by the table $(S, E)$.
	By Lemma~\ref{lem:hankel} this is bounded by the rank of the Hankel matrix.
	If $M$ and $N$ are modules generated by two tables such that $N$ is a strict quotient of $M$, then by Lemma~\ref{lem:quotrank} we have $\rank(M) > \rank(N)$.
	\qed
\end{proof}

Recall that, for termination of the algorithm, \Cref{thm:termination} requires a progress measure, which we defined above, and it requires the Hankel matrix of the language to satisfy the ascending chain condition (\Cref{def:ascending}). Proposition~\ref{prop:pidacc} shows that the latter is always the case for languages over PIDs. 
\begin{proposition}[Ascending chain condition PIDs]\label{prop:pidacc}
	The Hankel matrix of a language accepted by a WFA over a PID satisfies the ascending chain condition.
\end{proposition}
\begin{proof}
	Let $H$ be the Hankel matrix, which has finite rank by Lemma~\ref{lem:hankel}.
	If
	$$
		M_1 \subseteq M_2 \subseteq M_3 \subseteq \cdots
	$$
	is any chain of submodules of $H$, then
	$
		M = \bigcup_{i \in \N} M_i
	$
	is a submodule of $H$ and therefore also of finite rank~\cite{jacobson2012}.
	Let $B$ be a finite basis of $M$.
	There exists $n \in \N$ such that $B \subseteq M_n$, so $M_n = M$.
	\qed
\end{proof}

The last ingredient for the abstract algorithm is solvability of the semiring: the following fact provides a sufficient condition for a PID to be solvable.

\begin{proposition}[PID solvability]\label{prop:pidsolvable}
     A PID $\PID$  is solvable if all of its ring operations are computable and 
     if each element of $\PID$ can be {\em effectively} factorised into irreducible elements. 
\end{proposition}
\begin{proof}
It is well-known that a system of equations  of the form   $A \underline{x} = \underline{b}$
with integer coefficients can be efficiently solved via computing the Smith normal form~\cite{smit61:onsy} of $A$. 
	The algorithm generalises to principal ideal domains, if we assume that the factorisation 
	of any given element of the principal ideal domain\footnote{Note that factorisations exist as each principal ideal domain 
	is also a unique factorisation domain, cf.~e.g.~\cite[Thm.~2.23]{jacobson2012}.}
	into irreducible elements is computable, cf. the algorithm in~\cite[p.~79-84]{jacobson1953}. To see that all steps in this algorithm
	can be computed, one has to keep in mind that the factorisation can be used to determine the greatest common divisor
	of any two elements of the principal ideal domain.
	\qed
\end{proof}
\begin{remark}
 In the case that we are dealing with an Euclidean domain $\PID$, a sufficient condition for $\PID$ to be solvable is that Euclidean division is computable (again this can be deduced from inspecting the algorithm in~\cite[p.~79-84]{jacobson1953}). 
 Such a PID behaves essentially like the ring of integers.
\end{remark}

Putting everything together, we obtain the main result of this section.

\begin{theorem}[Termination for PIDs]\label{thm:pids}
	Algorithm~\ref{alg:main} can be instantiated and terminates for any language accepted by a WFA over a PID of which all ring operations are computable and of which each element can be effectively factorised into irreducible elements.
\end{theorem}
\begin{proof}
	To instantiate the algorithm, we need a closedness strategy.
	According to Proposition~\ref{prop:strategy} it is sufficient for the PID to be solvable, which is shown by Proposition~\ref{prop:pidsolvable}.
	Proposition~\ref{prop:pidprogress} provides a progress measure, and we know from Proposition~\ref{prop:pidacc} that the Hankel matrix satisfies the ascending chain condition, so by Theorem~\ref{thm:termination} the algorithm terminates.
	\qed
\end{proof}

The example run given in Section~\ref{sec:execution_example_original} is the same when performed over the integers.
We note that if the teacher holds an automaton model of the correct language, equivalence queries are decidable by lifting the embedding of the PID into its \emph{quotient field} to the level of WFAs and checking equivalence there.

\section{Discussion}\label{sec:discussion}

We have introduced a general algorithm for learning WFAs over arbitrary semirings, together with sufficient conditions for termination.
We have shown an inherent termination issue over the natural numbers and proved termination for PIDs.
Our work extends the results by Bergadano and Varricchio~\cite{Bergadano}, who showed that WFAs over fields could be learned from a teacher. 
Although we note that a PID can be embedded into its corresponding field of fractions, the WFAs produced when learning over the field potentially have weights outside the PID.

Algorithmic issues with WFAs over arbitrary semirings have been identified before.
For instance, Krob~\cite{krob1994} showed that language equivalence is undecidable for WFAs over the tropical semiring.

On the technical level, a variation on WFAs is given by probabilistic automata, where transitions point to convex rather than linear combinations of states.
One easily adapts the example from Section~\ref{sec:naturals} to show that learning probabilistic automata has a similar termination issue.
On the positive side, Tappler et al.~\cite{TapplerA0EL19} have shown that deterministic MDPs can be learned using an \LStar\ 
based algorithm. The deterministic MDPs in {\em loc.cit.} are very different from the automata in our paper, as their states generate observable output
that allows to identify the current state based on the generated input-output sequence. 

One drawback of the ascending chain condition on the Hankel matrix is that this does not give any indication of the number of steps the algorithm requires.
Indeed, the submodule chains traversed, although converging, may be arbitrarily long.
We would like to measure and bound the progress made when fixing closedness defects, but this turns out to be challenging for PIDs.
The rank of the module generated by the table may not increase.
We leave an investigation of alternative measures to future work.

We would also like to adapt the algorithm so that for PIDs it always produces minimal automata.
At the moment this is already the case for fields,\footnote{%
 There is one exception: the language that assigns 0 to every word,
 which is accepted by a WFA with no states. The algorithm initialises the set of row labels,
 which constitute the state space of the hypothesis, with the empty word. 
%	This second condition exists because we initialise the set of row labels, which constitute the state space of the hypothesis, with the empty word; the empty language can be accepted by a WFA with no states.
} since adding a row due to a closedness defect preserves linear independence of the image of $\row$.
For PIDs things are more complicated---adding rows towards closedness may break linear independence and thus a basis needs to be found in $\row^\sharp$.
This complicates the construction of the hypothesis.

Our results show that, on the one hand, WFAs can be learned over finite semirings and arbitrary PIDs (assuming computability of the relevant operations) and, on the other hand, that there exists an infinite commutative semiring for which they cannot be learned.
However, there are many classes of semirings in between commutative semirings and PIDs, of which we would like to know whether their WFAs can be learned by our general algorithm.

Finally, we would like to generalise our results to extend the framework introduced in~\cite{vanheerdt2017}, which focusses on learning automata with side-effects over a monad.
WFAs as considered in the present paper are an instance of those, where the monad is the free semimodule monad $\free(-)$.
At the moment, the results in~\cite{vanheerdt2017} apply to a monad that preserves finite sets, but much of our general WFA learning algorithm and termination argument can be extended to that setting.
It would be interesting to see if crucial properties of PIDs that lead to a progress measure and to satisfying the ascending chain condition could also be translated to the monad level.

\paragraph{Acknowledgments.} We thank Joshua Moerman for comments and discussions. 

\pagestyle{plain}

\bibliographystyle{plain}
\bibliography{main}

%%%%% To display Open Access text and logo, Please add below text and copy the cc_by_4-0.eps in the manuscript package %%%

%\vfill

{\small\medskip\noindent{\bf Open Access} This chapter is licensed under the terms of the Creative Commons\break Attribution 4.0 International License (\url{http://creativecommons.org/licenses/by/4.0/}), which permits use, sharing, adaptation, distribution and reproduction in any medium or format, as long as you give appropriate credit to the original author(s) and the source, provide a link to the Creative Commons license and indicate if changes were made.}

{\small \spaceskip .28em plus .1em minus .1em The images or other third party material in this chapter are included in the chapter's Creative Commons license, unless indicated otherwise in a credit line to the material.~If material is not included in the chapter's Creative Commons license and your intended\break use is not permitted by statutory regulation or exceeds the permitted use, you will need to obtain permission directly from the copyright holder.}

\medskip\noindent\includegraphics{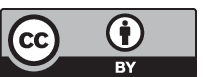}

\end{document}